\documentclass{llncs}
\usepackage[utf8]{inputenc}

\usepackage{subfig}
\usepackage{booktabs}
\usepackage{setspace}
\usepackage{epsfig}  
\usepackage{amsmath}  
\usepackage{url}  
\usepackage{lscape}  
\usepackage[justification=raggedright]{caption}	
\usepackage{amsfonts}
\usepackage{listings}
\usepackage{paralist}
\usepackage[ruled,commentsnumbered,linesnumbered,noend]{algorithm2e}
\usepackage{url}
\usepackage{tikz}
\usepackage[utf8]{inputenc}
\usetikzlibrary{arrows,automata}
\usetikzlibrary{positioning}
\usepackage{multicol}
\usepackage{multirow}
\usepackage{makecell}
\usepackage{scalerel}
\usepackage{xcolor}
\usepackage{subfig}
\usepackage{xspace}
\usepackage{todonotes}
\usepackage{graphicx}
\usepackage{enumitem}
\usepackage{wrapfig}

\graphicspath{{figures/}}
\numberwithin{equation}{section}

\title{Statistical Verification of Autonomous Systems using Surrogate Models and Conformal Inference}

\author{Chuchu Fan\inst{1}, Xin Qin \inst{2}, Yuan Xia \inst{2} , Aditya Zutshi\inst{3}, Jyotirmoy V. Deshmukh \inst{2}}

\institute{Massachusetts Institute of Technology  \and University of Southern California \and Galois Inc.} 

\newcommand{\reals}{\mathbb{R}}

\newcommand{\nnreals}{\mathbb R^{\geq 0}}

\newcommand{\restr}{\mathrel{\downarrow}}

\newcommand{\reach}[2]{\relax\ifmmode {\sf Reach}_{#1}(#2) \else ${\sf Reach}_{#1}(#2)$\fi} 
\newcommand{\reachhs}[1]{\relax\ifmmode {\sf Reach}_{#1} \else ${\sf Reach}_{#1}$\fi} 

\newcommand{\traj}{\xi}
\newcommand{\trajset}{\Xi}

\newcommand{\dom}{{\mathsf{dom}}}

\newcommand{\execs}[1]{\relax\ifmmode {\sf Execs}(#1) \else ${\sf Exec}(#1)$\fi} 
\newcommand{\traces}[1]{\relax\ifmmode {\sf Traces}(#1) \else ${\sf Tracec}(#1)$\fi}

\DeclareMathOperator*{\argmin}{arg\,min}
\newcommand{\simulator}{{\mathsf{sim}}}

\newcommand{\Lmode}[1]{\relax\ifmmode {\sf {#1}} \else ${\sf {#1}}$\fi}

\DeclareRobustCommand*\cal{\@fontswitch\relax\mathcal}

\newcommand{\states}{\mathcal{X}}

\newcommand{\num}[1]{\relax\ifmmode \mathbb #1\else $\mathbb #1$\fi}

\newcommand{\confint}{\mathtt{ConfInt}}
\newcommand{\paramspace}{\Theta}
\newcommand{\always}{\mathbf{G}}
\newcommand{\until}{\mathbf{U}}
\newcommand{\eventually}{\mathbf{F}}
\newcommand{\rob}{\rho}
\newcommand{\distribution}{\mathcal{D}}
\newcommand{\regpara}{\eta}
\newcommand{\pdf}{f}
\newcommand{\param}{\theta}
\newcommand{\simint}{\overline{\mathsf{sim}}}
\newcommand{\regalg}{\mathtt{Reg}}

\newcommand{\matlab}{Matlab\textsuperscript{\textregistered}\xspace}

\newcommand{\initPos}{x_{\mathrm{init}}}
\newcommand{\initVel}{v_{\mathrm{init}}}

\newcommand{\myipara}[1]{\vspace{0.2em} \noindent{\em #1}.}
\newcommand{\mypara}[1]{\vspace{0.2em} \noindent\textbf{#1}.}

\newcommand{\specmc}{\varphi_\textsc{mc}}
\newcommand{\speclkasettle}{\varphi_{\textsc{lka},\mathrm{settle}}}
\newcommand{\speclkabounds}{\varphi_{\textsc{lka},\mathrm{bounds}}}

\newcommand{\oldstuff}[1]{}

\newcommand{\model}{\mathcal{M}}
\newcommand{\surrogate}{\hat{\mu}}
\newcommand{\prob}{P}
\newcommand{\expectedvalue}{\mathbb{E}}
\newcommand{\Reg}{\mathcal{R}}
\newcommand{\Idx}{\mathcal{I}}

\newcommand{\sampleParams}{\widehat{\paramspace}}

\newcommand{\jointdist}{\distribution_{\param,\rob(\varphi)}}

\newcommand{\vmin}{v_{\min}}
\newcommand{\vmax}{v_{\max}}
\newcommand{\vstarmin}{v^*_{\min}}
\newcommand{\vstarmax}{v^*_{\max}}

\newcommand{\mindia}{\Delta}

\begin{document}
\maketitle

\begin{abstract} 
In this paper, we propose conformal inference based approach for statistical verification of CPS models.
Cyber-physical systems (CPS) such as autonomous vehicles, avionic
systems, and medical devices operate in highly uncertain environments.
This uncertainty is typically modeled using a finite number of parameters or input signals. Given a system
specification in Signal Temporal Logic (STL), we would like to verify that for all (infinite) values of the model
parameters/input signals, the system satisfies its specification.
Unfortunately, this problem is undecidable in general. {\em Statistical model checking} (SMC) offers a solution by providing guarantees on the correctness of CPS models by statistically reasoning on model simulations. 
We propose a new approach for statistical verification of
CPS models for user-provided distribution on the model parameters. Our
technique uses model simulations to learn {\em surrogate models}, and
uses {\em conformal inference} to provide probabilistic
guarantees on the satisfaction of a given STL property. 
Additionally,
we can provide prediction intervals containing the the quantitative 
satisfaction values of the given STL property for any user-specified 
confidence level. 
We also propose a refinement
procedure based on Gaussian Process (GP)-based surrogate models for obtaining fine-grained probabilistic guarantees over sub-regions in
the parameter space. This in turn enables the CPS designer to choose assured
validity domains in the parameter space for safety-critical
applications. Finally, we demonstrate the efficacy of our technique on several
CPS models.

\end{abstract}

\section{Introduction}\label{sec:intro}
Cyber-physical systems (CPS) such as robotic ground and aerial
vehicles, medical devices and industrial controllers are highly
complex systems with nonlinear behaviors that operate in uncertain operating
environments. As these systems are often safety-critical, it is desirable to obtain strong assurances on their safe operation.
To achieve this goal, recent research has been focused on effective
and sound verification
algorithms~\cite{verifyAI,Verisig,Reluplex,Yasser,Taylor,dutta2018learning,huang2017safety,zareistatistical},
and scalable best-effort approaches which lack explicit coverage
guarantees~\cite{hoxha2014towards,mooney1997monte,yaghoubi2019gray}.
However, factors like complexity and stochasticity of the operating
environments, curse of dimensionality, the nonlinearity of dynamics
pose a significant scalability challenge for verification procedures.

In this paper, we address the problem of analyzing the effects of uncertainty in
the environment on the correctness of a given CPS model $\model$.
We assume that the aleatoric uncertainty in
the environment is modeled as discrete-time input signals to $\model$,
where the signal value at each discrete time-step is assumed to come
from some distribution. In addition, we also consider {\em static}
parameters of $\model$, i.e. parameters that have a fixed value during
a simulation run; examples include design parameters of $\model$ or
parameters modeling variability in the sensor measurements and
actuation errors. We can effectively lump these different parameters
into a single multivariate random variable $\param$ that takes values
from some compact set $\paramspace$, distributed according to some
user-provided distribution $\distribution_\paramspace$.  For any
single sample of $\param$, we assume that the output trajectories of
the model (denoted $\traj_\param$) are deterministic, i.e. the model
is free of any {\em internal} stochastic behavior.  A distribution on
the parameters induces a distribution on the output trajectories
$\traj_\param$ of the model. 

The correctness of CPS models can be expressed using Signal Temporal
Logic (STL) \cite{maler2004} formulas over the model trajectories. STL
formulas are evaluated on trajectories, i.e. signals, and in this
paper we assume signals to be finite sequences of time-value pairs.
STL formulas have both Boolean and quantitative satisfaction semantics. The Boolean semantics are interpreted as follows: a signal $x$ either satisfies a formula $\varphi$ (denoted $x \models \varphi$) or it does not satisfy $\varphi$ (denoted $x \not\models \varphi$). The quantitative semantics are defined using a {\em robustness} value
$\rob(\varphi,x)$ \cite{donze_robust_2010,fainekos_robustness_2009}.
Intuitively, $\rob(\varphi,x)$ gives a degree of satisfaction of
$\varphi$ by $x$.  In this paper we are interested in answering the
following two questions:
\begin{enumerate}[itemsep=1pt,leftmargin=1.3em]
\item

For some user-provided threshold $\epsilon$, and $\param \sim
\distribution_\paramspace$, is the probability of the model
behavior satisfying a given STL property $\varphi$ more than
$1-\epsilon$, or 
\begin{align}
\label{eq:probsatcheck}
\param \sim \distribution_\paramspace \implies \prob(\model(\param) \models \varphi) \ge 1-\epsilon
\end{align}

\item

For some user-provided threshold $\epsilon$, and $\param \sim
\distribution_\paramspace$, can we find an interval $[\ell,u]$ s.t. 
the probability that the robustness value of a model behavior
$\model(\param)$ w.r.t. a given STL property $\varphi$ lies in
$[\ell,u]$ with probability greater than $1-\epsilon$? I.e.,
\begin{align}
\label{eq:probrobcheck}
\param \sim \distribution_\paramspace \implies
\prob(\rob(\varphi,\model(\param)) \in [\ell,u]) > 1-\epsilon
\end{align}
\end{enumerate}

Statistical model checking (SMC)
\cite{legay2015statistical,agha2018,zuliani2010bayesian,zareistatistical,NimaPowerTrain,MaheshSMC}
approaches have been used in the past to establish assertions such as
\eqref{eq:probsatcheck}. The most popular SMC methods use statistical
hypothesis testing procedures to check whether the hypothesis that
\eqref{eq:probsatcheck} is true can be accepted with confidence
exceeding user-specified thresholds $\alpha,\beta$ for respectively
committing a type I error (i.e. accepting the hypothesis when it is
not true), or a type II error( i.e. rejecting the hypothesis when it
is true).  SMC methods provide the user with conditions on the number
of simulations required, $\alpha$, $\beta$ and $\epsilon$ in order to
accept or reject the hypothesis~\eqref{eq:probsatcheck}.

\myipara{Approach} To establish assertions such as \eqref{eq:probsatcheck} or \eqref{eq:probrobcheck}, we present an approach based on conformal inference, a technique for giving confidence intervals with
marginal coverage guarantees. A unique feature of our technique is
that it does not make any assumptions on the user-provided
distribution on the parameter space or the dynamics represented by the
model.

We now give an overview of our approach as illustrated in
Fig.~\ref{fig:overview}.  Given the parameterized system model
$\model(\param)$, a distribution $\distribution_\paramspace$ over the
space $\paramspace$ of parameter values, and an STL property
$\varphi$, we perform $N$ simulations model with these parameter
values to obtain trajectories $\traj_\param$.  We then compute the
robust satisfaction value $\rob(\varphi,\traj_\param)$ for each model
trajectory. 

\begin{wrapfigure}{r}{0.5\textwidth}
\vspace{-0.5cm}
\includegraphics[width=.49\textwidth]{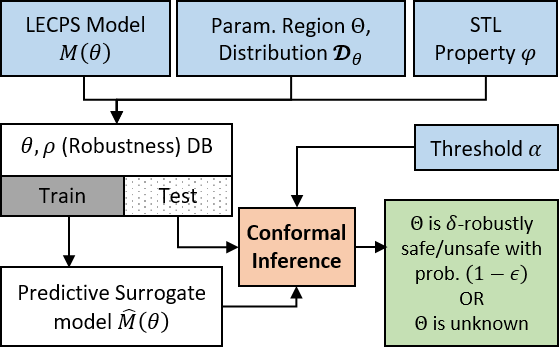}
\caption{Overview of our approach.\label{fig:overview}}
\end{wrapfigure}
Next, we consider some subset of the generated $\theta$
and corresponding $\rob(\varphi,\traj_\param)$ values and train a
surrogate model $\surrogate$ that takes in a parameter value $\param$
and predicts a robustness value $\surrogate(\param)$ for {\em all}
parameter values in the given region $\paramspace$. Finally, we use
the test set and a user-provided threshold $\epsilon$ as inputs to a
combination of the {\em conformal inference} procedure with a global
optimizer to compute a bound $d$ that guarantees that for all $\param'
\in \paramspace$, 
\[
\prob\left(\rob(\varphi,\traj_{\param'}) \in \left[\surrogate(\param') -d,\ \
                           \surrogate(\param')+d\right]\right) \geq 1-\epsilon.  
\]

The above guarantee allows us to construct a confidence interval for
the robustness of the property $\varphi$ over the given region
$\paramspace$.  A strictly positive or strictly negative interval
indicates that the $\paramspace$ is respectively safe or unsafe.
However, if the interval contains $0$, then the status of
$\paramspace$ remains unknown.

The above procedure naturally yields a refinement procedure which
allows us to start with a larger region in the parameter space, and
split it into smaller regions if the region is deemed unknown. In a
smaller region, the accuracy of the surrogate model improves (due to
more data in a smaller region), and hence previously inconclusive
regions can be resolved as safe/unsafe. A na\"{i}ve version of this
splitting algorithm faces the curse of dimensionality -- if the
parameter space is high-dimensional then the branch-and-bound
procedure ends up creating too many branches which can make the
procedure intractable.  In order to address this shortcoming, we
investigate a procedure that uses a Gaussian Process based surrogate
model combined with Bayesian optimization to split the parameter
regions in a smart manner. Gaussian processes explicitly encode sample
uncertainty, which allows us to perform splitting that adapts to the
robustness surface, thereby reducing the number of regions that the
algorithm examines.

Our method can scale to CPS models encoding complex dynamics and large
state spaces, as well as reasonably large parameter spaces. The
results of our method can be used to characterize safe operating
regions in the parameter space, and to build (probabilistic) safety
assurance cases. With respect to analysis times, our method compares
favorably with approaches based on statistical model checking (SMC).
However, unlike SMC and PAC based methods the provided guarantee is
not a function of the number of samples. In fact, our method can
potentially provide the needed level of probabilistic guarantees with
{\em any number of samples}. This is because we build a surrogate
model from samples; if the surrogate model is of poor accuracy due to
a limited number of samples, conformal inference will predict a wider
interval of robustness values {\em with the same level of guarantees}
$1-\alpha$, while for a more accurate model, the robustness interval
will be narrower. Thus, conformal inference allows a trade-off between
sample complexity and the tightness of the guarantee {\em independent}
of the level of the guarantee itself.

The main contributions of this paper are:
\begin{enumerate}
    \item 
        A technique based on surrogate models to approximate the
        robustness of a given specification learned using off-the-shelf
        regression techniques.
    \item
        A new technique for generating prediction intervals for
        the robustness of a specification with user-specified
        probabilistic thresholds.  
    \item
        Algorithms to partition the parameter space of a model
        into safe, unsafe and unknown regions based on conformal
        inference on the surrogate models.
    \item
        Experimental validation on CPS models demonstrating the
        real-world applicability of our methods.
\end{enumerate}

The rest of the paper is organized as follows.
Section~\ref{sec:prelim} provides the background and notation.
Section~\ref{sec:conformal} explains how we use conformal inference
for providing probabilistic guarantees on satisfaction/violation of a
given STL property over a region in the parameter space.
Section~\ref{sec:refinement} presents our algorithm for refining
parameter spaces using Gaussian Processes. Finally,
Section~\ref{sec:casestudies}, illustrates our approach using several
case studies and Section~\ref{sec:relconc} presents our conclusions.

\section{Preliminaries}\label{sec:prelim}

\begin{definition}[Signals, Black-box Models]
We define a signal or a trajectory $\traj$  as a function from a
set $\dom = [0,T]$ for some $T \in \reals^{\ge 0}$ to a compact 
set of values $\states$. The signal value at time $t$ is denoted as 
$\traj(t)$. A parameter space $\paramspace$ is some compact subset 
of $\reals^k$. A model $\model$ is a function that maps a parameter 
value $\param \in \paramspace$ to an output signal $\traj_\param$.
\end{definition}

We note that the above definition permits parameterized {\em input
signals} for the model. We can define such signals using a function
known as a signal generator that maps specific parameter values to
signals. For example, a piecewise linear signal containing $k$ linear
segments can be described using $k+1$ parameters, $k$ corresponding to
the starting point for each segment and $1$ for the end-point of the
final segment.

We  assume that $\param \in \paramspace$ is a random variable that
follows a (truncated) distribution $\distribution_\param$ with
probability density function (PDF) $\pdf(\param)$ and $\forall \param
\notin \paramspace, \pdf(\theta) = 0$.  If we only wish to draw
samples from a subset $S \subseteq \paramspace$ (by dropping samples
from $\paramspace\setminus S$), the corresponding distribution of the
samples is denoted by $\distribution_\param \restr S$ and follows the
PDF shown below.
\begin{equation}
\label{eq:pdfrestr}
  \pdf^\prime(\param) = \left\{
\begin{array}{ll}
     \frac{\pdf(\param)}{\int_{\tau \in S}\pdf(\tau)d\tau}& \mbox{if~} \param \in S \\
     0 &  \mbox{otherwise}.
\end{array}
\right.  
\end{equation}
Instead of closed form descriptions of the generator for
$\traj_\param$ (e.g. differential or difference equations),
we assume that there is a {\em simulator} that can generate signals
compatible with the semantics of the model $\model$. 
	
\begin{definition}
\label{def:sims}
A {simulator\/} for a (deterministic) set $\trajset$ of trajectories 
is a function (or a program) $\simulator$ that takes as input a 
parameter $\param \in \paramspace$, and a finite sequence of time points 
$t_0, \ldots, t_k$, and returns the signal
($t_0,\simulator(\param,t_0)$, $\ldots$, $t_k,\simulator(\param,t_k)$),
where for each $i\in \{0,\ldots, k\}$, $\simulator(\param,t_i) = 
\traj_\param(t_i)$.
\end{definition}
In rest of the paper, unless otherwise specified, we ignore the
distinction between the signals $\simulator(\param,\cdot)$ and
$\traj_\param$.

%


\subsection{Signal Temporal Logic}
\label{sec:STL}

Signal Temporal Logic \cite{maler2004} is a popular formalism that has
been used widely used to express safety specifications for many CPS
applications. STL formulas are defined over signal predicates of the
form $f(\traj) \geq c$ or $f(\traj) \leq c$, where $\traj$ is a signal
and $f:\reals^n \rightarrow \reals$ is a real-valued function and $c
\in \reals$. STL formulas are written using the grammar shown in
Eq.~\eqref{eq:stlsyntax}. Here, we assume that $I=[a,b]$, where
$a,b\in\reals^{\ge 0}, a \leq b$, and $\sim \in \{\leq,\geq\}$.
\begin{equation}
\label{eq:stlsyntax}
\varphi,\psi:= \mathit{true} \mid
               f(\traj) \sim c \mid
               \neg \varphi \mid
               \varphi \wedge \psi \mid
               \varphi \vee \psi \mid
               \eventually_I \varphi \mid
               \always_I \varphi \mid
               \varphi~\until_I~\psi
\end{equation}
In the above syntax, $\eventually$ (eventually), $\always$ (always),
and $\until$ (until) are temporal operators.  Given $t \in \nnreals$
and $I = [a,b]$, we use $t + I$ to denote $[t+a, t+b]$.  Given a
signal $\traj$ and a time $t$, we use $(\traj, t) \models \varphi$ to
denote that $\traj$ satisfies $\varphi$ at time $t$, and $\traj
\models \varphi$ as short-hand for $(\traj, 0) \models \varphi$. The
Boolean satisfaction semantics of an STL formula can be recursively in
terms of the satisfaction of its subformulas over the a signal. For
$\sim \in \{\leq,\geq\}$, $\traj$: $(\traj,t) \models f(\traj)\sim c$
if $f(\traj(t)) \sim c$ is true. The semantics of the Boolean
operators for negation ($\neg$), conjunction ($\wedge$) and
disjunction ($\vee$) can be obtained in the usual fashion by applying
the operator to the Boolean satisfaction of its operand(s). The value
of $(\traj,t) \models \eventually_I\varphi$ is true iff $\exists t'\in
t+I$ s.t. $(\traj,t') \models \varphi$, while $(\traj,t)\models
\always_I\varphi$ iff $\forall t'\in t+I$, $(\traj,t') \models
\varphi$. The formula $\varphi\until_I\psi$ is satisfied at time $t$
if there exists a time $t'$ s.t. $\psi$ is true, and for all $t'' \in
[t,t')$, $\varphi$ is true.


STL is also equipped with quantitative semantics that define the {\em
robust  satisfaction value} or {\em robustness} -- a function mapping
a formula $\varphi$ and the signal $\traj$ to a real number
\cite{fainekos_robustness_2009,donze_robust_2010}. Informally,
robustness can be viewed as a degree of satisfaction of an STL formula
$\varphi$. While many competing definitions for robust satisfaction
value
exist~\cite{akazaki_time_2015,rodionova_temporal_2016,jaksic_quantitative_2018},
we use the original definitions~\cite{donze_robust_2010} in this
paper.  
\begin{definition}
\label{def:RSV}
The robustness value is a function $\rob$  mapping $\varphi$,
the trajectory $\traj$, and a time $t \in \traj.\dom$ as follows:
\[
\small
\begin{array}{rcl}
    \rob(f(\traj)\ge c, \traj, t) & = & f(\traj(t)) - c  \\
    \rob(\neg \varphi, \traj, t) & = & -\rob(\varphi, \traj, t) \\ 
    \rob(\varphi \wedge \psi, \traj, t) &=& \min( \rob(\varphi, \traj, t), \rob(\psi, \traj, t)) \\
    \rob(\varphi ~\until_{I}~ \psi) &=& \sup\limits_{t_1 \in t+I} \min 
    (\rob (\psi, \traj, t_1), \inf\limits_{t_2 \in [t,t_1)} \rob(\varphi, \traj, t_2) ) \\
\end{array}
\]
\end{definition}
The robustness values for other Boolean and temporal operators can be
derived from the above definition; for example, $\always_I \varphi$ and
$\eventually_I \varphi$ are a special case of the semantics for until
($\until_I$) respectively evaluating to the minimum and maximum of the
robustness of $\varphi$ over the interval $I$.

\begin{example} 
\label{example:one}
Consider the time-reversed van Der Pol oscillator specified as
$\dot{x_1} = -x_2$, $\dot{x_2} = 4(x_1^2-1)x_2 + x_1$.
Figure~\ref{figure:STLandTraj} illustrates the satisfaction (indicated
in blue) and violation (indicated in red) of two example
specifications by $x_1(t)$: (a) $\varphi_1$ specifies that for any
time $t \in [0, 10]$, the value of the trajectory $x(t)$ should be
less than $0.5$ and (b) $\varphi_2$ specifies that from some time
within the first $2$ time units, $x(t)$ settles in the region
$[-0.3,0.3]$ for $8$ time units.
\end{example}

\begin{figure}[t]
\centering
\subfloat[$\varphi_1  = \always_{[0,10]} (x(t) < 0.5)$\label{fig_phi_1}]{%
\includegraphics[width=0.49\textwidth,trim={0.2cm 0 9cm 0},clip]{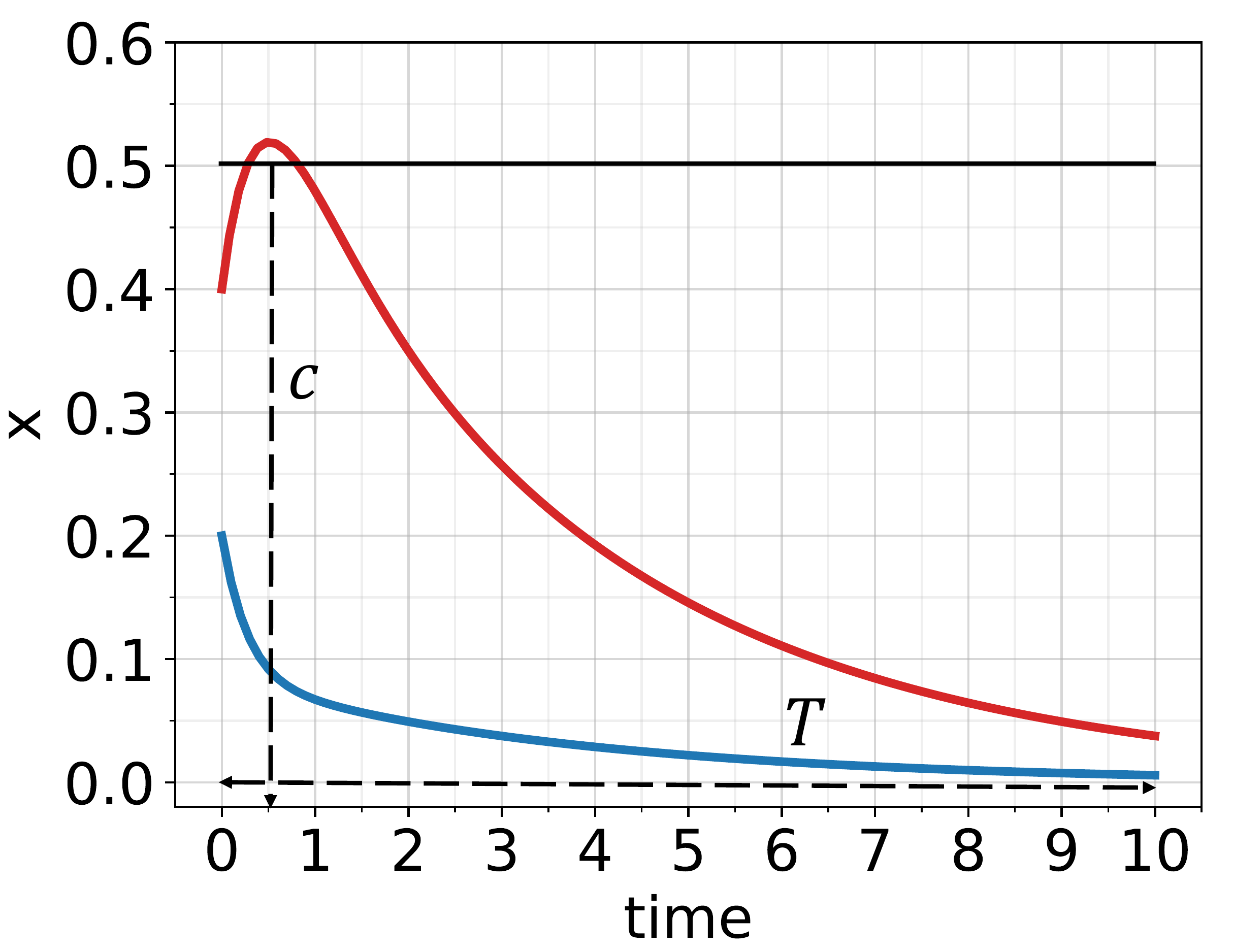}
}%
\subfloat[$\varphi_2 = \eventually_{[0,2]} \always_{[0,8]} (\| x(t)\| < 0.3)$%
\label{fig_phi_2}]{%
\includegraphics[width=0.49\textwidth,
                 trim={0.2cm 0 9cm 0},clip]{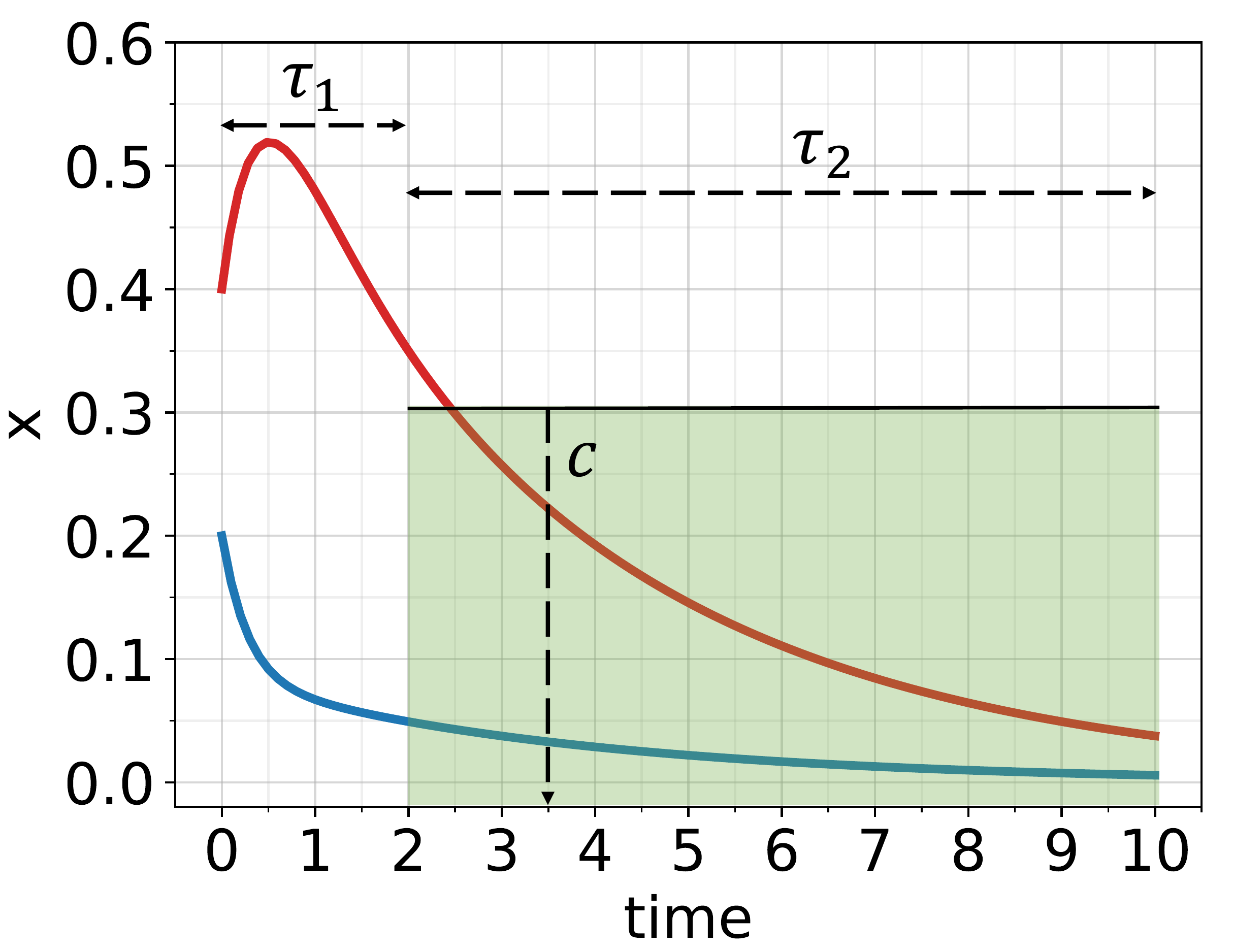}}
\caption{Example of trajectories satisfying (blue) and violating (red) the 
         STL formulas.\label{figure:STLandTraj}}
\end{figure}



\subsection{Learning Surrogate Models}\label{sec:learning}
In this section we discuss learning of surrogate models for a given
black-box model $\model$. A surrogate model is essentially a {\em
quantitative abstraction} of the original black-box model.
Quantitative abstractions have been explored
in the theory of weighted transition systems (WTS)
\cite{cerny2013quantitative}. A WTS is a transition system where every
transition is associated with weights, and a quantitative property of
the WTS maps sequences of states of the WTS to a real number computed
using some arithmetic operations on the weights. Quantitative
abstractions focus on sound proofs for quantitative properties. We
observe that we can view the robustness of an STL property as a
quantitative property evaluated on the system trajectory. We introduce
two new notions of quantitative abstractions defined on the
trajectories of a system.

\begin{definition}[$\delta$-surrogate model]
Let $\traj_\param$ be the trajectory obtained by simulating the
$\model$ with the parameter $\param$, where $\param \in \paramspace$.
Let $\gamma$ be a {\em quantitative property} on $\traj_\param$, i.e.
$\gamma$ maps $\traj_\param$ to a real number. We say that a model
$\surrogate$ that maps $\param$ to a real number is an
$\delta$-distance-preserving quantitative abstraction or an
$\delta$-surrogate model of $\model$ and $\gamma$ if 
\begin{align}
\label{eq:epsdistpresabs}
\exists \delta: \forall \param \in \paramspace: |\gamma(\traj_\param) - \surrogate(\param)| \le \delta
\end{align}
\end{definition}

Essentially, the $\delta$-surrogate model guarantees that the value
of the quantitative property $\gamma$ evaluated on $\traj_\param$
(obtained from the original model $\model$) is no more than $\delta$
away from the value that it predicts. The idea is that the
$\delta$-surrogate model could be systematically derived from the
original model, and could be significantly simpler than the original
model making it amenable to formal analysis.  For example, if we have
an $\delta$-surrogate model, then we can prove that a given property
holds by systematically sampling the parameter space $\paramspace$.

In general, such models could be hard to obtain; hence, we propose a
probabilistic relaxation known as the $(\delta,\epsilon)$-{\em probabilistic
surrogate} model, where condition~\eqref{eq:epsdistpresabs}
is replaced by \eqref{eq:probespdistpresabs}.
\begin{definition}[$(\delta,\epsilon)$-probabilistic surrogate model]
Given a model $\model$, a quantitative property $\gamma$, and a
user-specified bound $\epsilon \in [0,1)$, we say that
$\surrogate$ is a $(\delta,\epsilon)$-probabilistic surrogate model if:
\begin{align}
\label{eq:probespdistpresabs}
\exists \delta \in \reals, \epsilon \in [0,1):
\prob\left(|\gamma(\traj_\param) - \surrogate(\param)| \le \delta \mid \param \sim \distribution_\param \right) \ge 1-\epsilon
\end{align}
\end{definition}

We now explain how we can obtain $(\delta,\epsilon)$-probabilistic
surrogate models for an arbitrary quantitative property $\gamma$. The
basic idea is to use statistical learning techniques: we sample
$\paramspace$ in accordance with the distribution
$\distribution_\param$ to obtain a finite set of parameter values
$\sampleParams$. For each $\param_i \in \sampleParams$, we simulate
the model to obtain $\traj_{\param_i}$ and compute
$\gamma(\traj_{\param_i})$. We then compute the surrogate model
$\surrogate$ using parametric regression models (e.g. linear,
polynomial functions) or nonparametric regression methods (e.g. neural
networks and Gaussian
Processes)~\cite{friedman2001elements,barron1993universal,gpbook}.  We
now briefly review some of these regression methods, and in
Section~\ref{sec:conformal} explain how we can obtain $\delta$ values
for a user-provided bound $\epsilon$.

\myipara{Polynomial Regression}
Polynomial regression assumes a polynomial relationship between independent
variables $X$ and the dependent variable $Y$. It aims to fit a polynomial curve
to the input and output data in a way that minimizes a suitable loss function.
A commonly used loss function is the least square error (or the sum of squares
of residuals). Typically, a polynomial regression requires the user to specify
the degree of the polynomial to use. Polynomial regression generally has high
tolerance to the function's curvature level, but has high sensitivity to the
outliers. In our experiments, we restrict the polynomial degree to 2.

\myipara{Neural Network Regression} Neural networks
\cite{bishop2006pattern} offer a high degree of flexibility for
regressing arbitrary nonlinear functions. While there are many
different NN architectures, we use a simple multi-layer perceptron
model with a stochastic gradient-based optimizer. This model simply
updates its parameters based on iterative steps along the partial
derivatives of the loss function. 

\myipara{Gaussian Process based Regression Model \cite{gpbook}} A
Gaussian Process (GP) is a stochastic process, i.e., it is a
collection of random variables $W_\theta$ indexed by $\theta$, where
$\theta$ ranges over some discrete or dense set. The key property of
GP is that any finite sub-collection of these random variables has a
multi-variate Gaussian distribution. GP models are popular as
non-parametric regression methods used for approximating arbitrary
continuous functions with the appropriate {\em kernel functions}. A GP
can be used to express a prior distribution on the space of functions,
e.g. from a domain $\reals^n$ to $\reals$. Let $F:\reals^n\to \reals$
be a random function. Then, we say that $F$ is a centered Gaussian
process with kernel $k$, if for every $(x_1,\ldots,x_n) \in \reals^n$,
there exists a positive semi-definite matrix $\Sigma$ such that
$[F(x_1),\ldots,F(x_n)]\sim (0,\Sigma)$. The $(i,j)^{th}$ entry of
$\Sigma$, i.e. $\Sigma_{ij}=k(x_i,x_j)$ for some kernel function $k$.
The matrix $\Sigma$ is called the covariance matrix, and the function
$k$ measures the joint variability of $x_i$ and $x_j$. There are
several kernel functions that are popular in literature: the squared
exponential kernel, the 5/2 Mat\'{e}rn kernel, etc. In our
experiments, we use a {\em sum} kernel function that is the addition
of a dot product kernel and a white noise kernel (explained in
Section~\ref{sec:refinement}).


\section{Conformal Inference}\label{sec:conformal}
Conformal inference~\cite{lei2018distribution,lei2014distribution} is
a framework to quantify the accuracy of predictions in a regression
framework~\cite{vovk2005algorithmic}. It can provide guarantees using
a finite number of samples, without making assumptions on the
distribution of data used for regression or the technique used for
regression. We explain the basic idea of conformal inference, and then
explain how we adapt it to our problem setting.

\subsection{Conformal Inference Recap}
Consider i.i.d. regression data $Z_1,\cdots,Z_m$ drawn from an
arbitrary joint $\distribution_{XY}$, where each $Z_i = (X_i,Y_i)$ is
a random variable in $\reals^n \times \reals$, consisting of
$n$-dimensional feature vectors $X_i$ and a response variable $Y_i$.
Suppose we fit a surrogate model to the data, and we now wish to use
this model to predict a new response $Y_{m+1}$ for a new feature value
$X_{m+1}$, with no assumptions on $\distribution_{XY}$. Formally,
given a positive value $\alpha \in (0,1)$, conformal inference
constructs a prediction band $C \subseteq \reals^n \times \reals$
based on $Z_1,\cdots,Z_n$ with
property~\eqref{eq:finite-sample-guarantee}.
\begin{equation}
\label{eq:finite-sample-guarantee}
    \prob(Y_{m+1} \in C(X_{m+1})) \geq 1-\alpha.
\end{equation}
\noindent Here, the probability is over $m+1$ i.i.d. draws
$Z_1,\cdots,Z_{m+1} \sim \distribution_{XY}$, and for a point $x \in
\reals^n$ we denote $C(x) = \{y \in \reals: (x, y)\in C\}$. The
parameter $\alpha$ is called the {\em miscoverage level} and
$1-\alpha$ is called the {\em probability threshold}. Let 
\[
\mu(x) = \expectedvalue(Y \ | \ X = x), x \in \reals^n
\]
denote the regression function, where $\expectedvalue(W)$ denotes the
expected value of the random variable $W$. The regression problem is
to estimate such a conditional mean of the test response $Y_{m+1}$
given the test feature $X_{m+1} = x$.  Common regression methods use a
regression model $g(x,\regpara)$ and minimize the sum of squared
residuals of such model on the $m$ training regression data
$Z_1,\cdots,Z_m$, where $\regpara$ are the parameters of the
regression model.  An estimator for $\mu$ is given by $\hat{\mu}(x) =
g(x,\hat{\regpara})$, where 
\[ 
\hat{\regpara} = \displaystyle\argmin_{\regpara} 
                  \frac{1}{m} \sum_{i=1}^{m} 
                    \left(Y_{i} - g(X_i,\regpara)\right)^2 + 
                    \Reg(\regpara)\]
and $\Reg(\regpara)$ is a regularizer.  In \cite{lei2018distribution},
the authors provide a technique called {\em split conformal
prediction} that we use to construct prediction intervals that satisfy
the finite-sample guarantees as in
Equation~\eqref{eq:finite-sample-guarantee}. The procedure is
described in Algorithm~\ref{alg:conformal} as a function $\confint$
which takes as input the i.i.d. training data $\{(X_i,
Y_i)\}_{i=1}^m$, miscoverage level $\alpha$ and any regression
algorithm $\regalg$. Algorithm~\ref{alg:conformal} begins by splitting
the training data into two equal-sized disjoint subsets. Then a
regression estimator $\hat{\mu}$ is fit to the training set $\{(X_i,
Y_i)\}: i \in \Idx_1)$ using the regression algorithm $\regalg$
(Line~\ref{alg1:fit}). Then the algorithm computes the absolute
residuals $R_i$s on the test set $\{(X_i, Y_i)\}: i \in \Idx_2)$
(Line~\ref{alg2:calibrate}).  For the desired probability threshold
$\alpha \in [0,1)$, the algorithm sorts the residuals in ascending
order $\{R_i: i\in \Idx_2\}$ and finds the residual at the position
given by the expression: $\left\lceil (n/2+1)(1-\alpha) \right\rceil$.
This residual is used as the confidence range $d$. In
\cite{lei2018distribution}, the authors prove that the prediction
interval at a new point $X_{m+1}$ is given by such $\hat{\mu}$ and $d$
that Theorem~\ref{thm:conformal} is valid.

\begin{algorithm}[tbhp!]
\caption{Conformal regression algorithm $\confint(\{(X_i, Y_i)\}_{i=1}^m, \alpha, \regalg)$}
\label{alg:conformal}
\SetKwInOut{Input}{input}
\SetKwInOut{Output}{output}
\Input{Data $\{(X_i, Y_i)\}_{i=1}^m$, miscoverage level $\alpha$, regression algorithm $\regalg$}
\Output{Regression estimator $\hat{\mu}$, confidence range $d$}
Randomly split $\{1,\cdots,n\}$ into two equal-sized subsets $\Idx_1, \Idx_2$\;
$\hat{\mu} = \regalg((X_i, Y_i): i \in \Idx_1)$ \label{alg1:fit} \ \;
$R_i = |Y_i - \hat{\mu}(X_i)|, i \in \Idx_2$ \label{alg2:calibrate}\;
$d = $ the $k$th smallest value in $\{R_i: i\in \Idx_2\}$, where $k=\lceil (n/2+1)(1-\alpha) \rceil$ \label{alg1:residue}\;
\Return $\hat{\mu}, d$
\end{algorithm}

\begin{theorem}[Theorem 2.1 in~\cite{lei2018distribution}]
\label{thm:conformal}
If $(X_i,Y_i),i=1,\cdots,m$ are i.i.d., then for an new i.i.d. draw $(X_{m+1}, Y_{m+1})$, using $\hat{\mu}$ and $d$ constructed in Algorithm~\ref{alg:conformal}, we have that
$
\prob(Y_{m+1} \in [\hat{\mu}(X_{m+1}) -d, \hat{\mu}(X_{m+1}) +d ] ) \geq 1 - \alpha.
$
Moreover, if we additionally assume that the residuals $\{R_i: i\in \Idx_0\}$ have a
continuous joint distribution, then
$
\prob(Y_{m+1} \in [\hat{\mu}(X_{m+1}) -d, \hat{\mu}(X_{m+1}) +d ] ) \leq 1 - \alpha +\frac{2}{m+2}.$
 \hfill $\square$
\end{theorem}

Generally speaking, as we improve our surrogate model $\hat{\mu}$ of
the underlying regression function $\mu$, the resulting conformal
prediction interval decreases in length. Intuitively, this happens
because a more accurate $\hat{\mu}$ leads to smaller residuals (or
$\epsilon$ in Section~\ref{sec:learning}), and conformal intervals are
essentially defined by the quantiles of the (augmented) residual
distribution.  Note that Theorem~\ref{thm:conformal} asserts marginal
coverage guarantees, which should be distinguished with the
conditional coverage guarantee $\prob(Y_{m+1} \in C(x) \ | \ X_{m+1} =
x) \geq 1-\alpha$ for all $x \in \reals^n$. The latter one is a much
stronger property and hard to be achieved without assumptions on
$\distribution_{XY}$.

\subsection{Computing $(\delta,\epsilon)$-probabilistic surrogate models}
\label{sec:para-to-rsv}



We assume that the parameter value $\param$ and $\rob(\varphi,
\traj_{\param})$ follow a joint (unknown) distribution $\jointdist$
that we wish to empirically estimate.  As indicated in
Section~\ref{sec:learning}, the first step to learning a
$(\delta,\epsilon)$-probabilistic surrogate model is based on sampling
$\jointdist$ and applying regression methods.  We draw $m$ i.i.d
samples $\sampleParams = \{\param_1,\cdots,\param_m\}$ from
$\distribution_\param$ and compute the robustness values $\rob_i =
\rob(\varphi,\traj_{\param_i})$ for each model trajectory
corresponding to the parameter $\param_i$.
Lemma~\ref{lemma:deltaepssurrogate} follows from
Theorem~\ref{thm:conformal}.
\begin{lemma}
\label{lemma:deltaepssurrogate}
Let $(\surrogate,d) = \confint(\{\param_i,\rob_i\}, \epsilon,
\regalg)$, where $\confint$ is as defined in
Algorithm~\ref{alg:conformal}, $1-\epsilon$ is a user-provided
probability threshold, $\regalg$ is some regression algorithm, and
$d \in \reals$, then $\surrogate$ is a $(d,\epsilon)$-probabilistic
surrogate model.
\end{lemma}

We now show how we can use $(\delta,\epsilon)$-probabilistic surrogate
models to perform statistical verification.
Theorem~\ref{thm:predict_interval} gives shows that the confidence
range returned by the conformal inference procedure can be extended
over the entire parameter space.
\begin{theorem}
\label{thm:predict_interval}
Let 
\begin{enumerate}
\item
$(\param_i,\rob_i),i=1,\cdots,m$ be i.i.d. samples drawn from the
joint distribution $\jointdist$ of $\param \in \paramspace$ and
$\rob(\varphi,\traj_\param)$,
\item
$\regalg$ be a regression algorithm,
\item
$1-\epsilon$ be a user-provided probability threshold,
\item
$(\surrogate,d) = \confint(\{\param_i,\rob_i\}, \epsilon, \regalg)$,
i.e. $\surrogate$ is the surrogate model and $d$ is the confidence range
returned by Algorithm~\ref{alg:conformal},
\item
$\vstarmax = \max_{\param \in \paramspace}
\surrogate(\param)$, and, $\vstarmin = \min_{\param \in \paramspace}$.
\end{enumerate}
Then, 
\begin{equation}
\prob\left(\rob(\varphi,\traj_\param) \in [\vstarmin - d, \vstarmax + d]\mid \param \sim \distribution_{\param} \right) \geq 1-\epsilon
\end{equation}
\end{theorem}

\begin{proof}
From Theorem~\ref{thm:conformal}, we know that {\em any new} i.i.d.
sample $(\param',\rob(\varphi,\traj_{\param'}))$ from $\jointdist$
satisfies:
\begin{equation}\label{eq:confguarantee} 
\prob(\rob(\varphi,\traj_{\param'}) \in [\surrogate(\param')-d, \surrogate(\param')+d]) \geq 1-\epsilon.  
\end{equation}
By definition, $\vstarmin \le \surrogate(\param') \le \vstarmax$.
Combining this with Eq.~\eqref{eq:confguarantee}, we get the desired
result.
\end{proof}

Theorem~\ref{thm:conformal} requires us to obtain the minimum/maximum
values of the surrogate model over a given region in the parameter
space.  If $\surrogate(\param)$ is a non-convex function and the
chosen optimization algorithm cannot compute the perfect optimal value
$\vstarmin$ or $\vstarmax$, but can only give conservative estimates
of the optimal value, we can update the predicted interval in
Theorem~\ref{thm:predict_interval} as follows.
\begin{corollary}
\label{co:adjusted_predict_interval}
Let $\vmin$ and $\vmax$ be respectively under- and over-approximations
of $\vstarmin$ and $\vstarmax$, then 
\begin{equation}
\label{eq:adjusted_constant_interval}
\prob\left(\rob(\varphi,\traj_\param) \in [\vmin-d, \vmin+d]\mid \param \sim \distribution_{\param} \right) \geq 1-\epsilon
\end{equation}
\end{corollary}
The bounds $\vmin$ and $\vmax$ in
Corollary~\ref{co:adjusted_predict_interval} can be computed using
global optimization solvers, SMT sovlers, or range analysis tools for
neural networks~\cite{Taylor,dutta2018learning}, based on the kind of
regression model used.


We can use the bounds obtained in Theorem~\ref{thm:predict_interval}
(similarly those with Corollary~\ref{co:adjusted_predict_interval}) to derive
probabilistic bounds on the Boolean satisfaction of a given STL
property $\varphi$, as expressed in
Theorem~\ref{thm:pure-sign-interval}.
\begin{theorem}
\label{thm:pure-sign-interval}
If $\vstarmin - d > 0$, then $\prob_{\distribution_\param}(\traj_\param \models \varphi \mid
\param \in \paramspace) \geq 1-\epsilon$.  
If $\vstarmax + d < 0$ then $\prob_{\distribution_\param}(\traj_\param \not\models \varphi
\mid \param \in \paramspace)
\geq 1-\epsilon$.
\end{theorem}
\begin{proof}
From \cite{fainekos_robustness_2009}, we know that
$\rob(\varphi,\traj_\param) > 0 \implies \traj_\param \models
\varphi$. Thus, if the lower bound of the prediction interval in
Theorem~\ref{thm:predict_interval} is positive, then $\traj_\param
\models \varphi$. The second case follows by a similar argument.
\end{proof}
If the first statement in the above theorem holds, we say that
$\paramspace$ {\em safe}, if the second statement holds, we say that
$\paramspace$ is {\em unsafe}, and if neither statement holds (i.e.
the predicted interval contains $0$), then we say that $\paramspace$
is {\em unknown}. While Theorem~\ref{thm:pure-sign-interval} allows us
a way to identify whether a region in the parameter space is is safe
(or unsafe), unfortunately there are two challenges: (1) the function
mapping $\param$ to $\rob(\varphi,\simint(\param))$ is a highly
nonlinear function in general, and {\em a priori} choice for a
regression algorithm $\regalg$ that fits this function with small
residual values may be difficult, (2) if there is large variation in
the value of the regression function over $\paramspace$, it is likely
that the conformal interval contains $0$, thereby marking
$\paramspace$ as unknown.  To circumvent this issue, one solution is
to split the parameter space $\paramspace$ into smaller regions where
it may be possible to get narrow conformal intervals at the same level
of probability threshold. We present a na\"{i}ve algorithm based on
parameter-space partitioning next. 

\subsection{Na\"{i}ve Parameter Space Partitioning} 
\label{sec:naive}

We now present an algorithm that uses
Theorem~\ref{thm:predict_interval} (or
Corollary~\ref{co:adjusted_predict_interval}) to provide probabilistic
guarantee by recursively splitting the parameter space $\paramspace$
into smaller regions such that each region can be labeled as safe,
unsafe or unknown. The basic idea of this algorithm is to compute the
conformal interval using Theorem~\ref{thm:predict_interval} and then
check if $\vstarmin-d <0$ and $\vstarmax+d >0$. If yes, we need to
partition the region. After partitioning the region, we have to repeat
the process of computing the conformal interval for each of the
sub-regions.  Note that the probability in
Theorem~\ref{thm:predict_interval} (and Theorem~\ref{thm:conformal} )
is marginal, being taken over all the i.i.d. samples $\{\param_i,
\rob_i\}$ from $\jointdist$. Therefore, when we work on each subset $S
\subseteq \paramspace$ after the partitions, we will have to restrict
$\param$ to be in $S$ (according to Equation~\eqref{eq:pdfrestr}) to
ensure that the Theorem~\ref{thm:predict_interval} is valid. We abuse
the notation and denote the joint distribution $\jointdist$ when
$\param$ is restricted to be sampled from $S \subseteq \paramspace$ by
$\jointdist \restr S$.

\begin{algorithm}[htb]
\caption{Parameter space partition with respect to STL formulas using conformal regression.}
\label{alg:main}
\SetKwInOut{Input}{input}
\SetKwInOut{Output}{output}
\Input{Parameter space $\paramspace$ and corresponding distribution
$\distribution_\param$, simulator $\simulator$ and interpolation
method to provide $\simint$, miscoverage level $\alpha$, regression
algorithm $\regalg$, an STL formula $\varphi$, a vector $\mindia$}
\Output{Parameter set $\paramspace^+$
that lead to satisfaction of $\varphi$, $\paramspace^-$ that lead to
violation of $\varphi$, and the rest parameter set $\paramspace^{U}$
that is undecided}
$\paramspace^+,\paramspace^-, \paramspace^{U} \gets \emptyset$, $\paramspace^{r} \gets \{ \paramspace \}$\;
\While{ $\paramspace^{r} \neq \emptyset$ \label{al:while}
}{ 	
$S \gets \mathtt{Pop}(\paramspace^{r})$ \;
$\param_1,\cdots, \param_m \gets \mathtt{IID\_Sample}(\distribution_\param\restr S)$ \;
\For{$i = 1, \cdots, m$}
{
$\rob_i \gets \rob(\varphi, \traj_{\param_i})$ \;
}
$\hat{\mu} ,d \gets \confint(\{(\param_i, \rob_i)\}_{i=1}^m, \alpha, \regalg)$\;
$v_{\max} \gets \max_{\param \in S}\hat{\mu}(\param)$, $v_{\min} \gets \min_{\param \in S}\hat{\mu}(\param)$\;
\uIf{$v_{\min}-d \geq 0$}{
$\paramspace^{+} \gets \paramspace^{+} \cup (S,[v_{\min}-d, v_{\max}+d])$ \label{alg2:positive}\;
}
\uElseIf{$v_{\max}+d \leq 0$}{
$\paramspace^{-} \gets \paramspace^{-} \cup (S,[v_{\min}-d, v_{\max}+d])$ \label{alg2:negative}\;
}
\uElseIf{$\mathtt{Diameters}(S) < \mindia\mathtt{Diameters}(\paramspace)$}{
$\paramspace^{U} \gets \paramspace^{U} \cup (S,[v_{\min}-d, v_{\max}+d])$ \label{alg2:unknown}\;
}
\Else{
$\paramspace^{r}.\mathtt{Push}(\mathtt{Partition}(S, \regalg))$ \label{alg2:partition}\;
}
}\label{al1:endwhile}
\Return $\paramspace^{+}, \paramspace^{-}, \paramspace^{U}$\;
\end{algorithm}

Algorithm~\ref{alg:main} searches over the parameter space
$\paramspace$ and partitions it to sets $\paramspace^{+}$,
$\paramspace^{-}$, and $\paramspace^{U}$, along with the prediction
intervals for the robustness values in each set. We first check if the
robustness value is strictly positive or negative and accordingly add
the region being inspected $S$ into $\paramspace^+$ or $\paramspace^-$
(Lines~\ref{alg2:positive} and~\ref{alg2:negative}).  When
Algorithm~\ref{alg:main} cannot decide whether $S$ belongs to
$\paramspace^{+}$ or $\paramspace^{-}$ the interval contains $0$, we
first check if for all $n$, the diameter of $S$ along the $n^{th}$
parameter dimension less than the fraction $\mindia_n
\mathtt{Diameters}(\paramspace)_n$. We assume that the vector
$\mindia$ is provided by the user.  If yes, the region is marked as
unknown. Otherwise, we partition $S$ into a number of subregions, that
are then added to a worklist of regions (Line~\ref{alg2:partition}).
In our implementation, in order to keep the number of subsets to be
explored bounded, we randomly pick a dimension in the parameter space,
and split the parameter space into two equal subsets along that
dimension.  Note that the partitioning can be accelerated by using
parallel computation, but we leave that for future exploration. For
each subset $S$, Algorithm~\ref{alg:main} additionally gives the
corresponding prediction interval, which indicates how good (or bad)
the trajectories satisfy (or violate) $\varphi$.  
\begin{theorem}
\label{thm:subsets}
In Algorithm~\ref{alg:conformal}, 
\(
\prob(\traj_\param \models \varphi \mid \param \sim \jointdist \restr \bigcup S \in \paramspace^{+}) \geq 1-\epsilon,
\), and
\(
\prob(\traj_\param \not\models \varphi \mid \param \sim \jointdist \restr \bigcup S \in \paramspace^{-}) \geq 1-\epsilon.
\)
\end{theorem}
Theorem~\ref{thm:subsets} directly follows from
Theorems~\ref{thm:predict_interval} and \ref{thm:pure-sign-interval}
and the total probability theorem.

\section{Gaussian Processes for Refinement}\label{sec:refinement}
A drawback of Algorithm~\ref{alg:main} is that the na\"{i}ve splitting
procedure is not scalable in high dimensions, and may have poor 
performance if the safe/unsafe regions have arbitrary shapes. In this
section, we instead suggest the use of Gaussian Processes (GP) coupled
with Bayesian updates to intelligently partition regions. Recall from
Section~\ref{sec:prelim}, for each parameter value $\param$, 
the GP model allows representing the mean $\mu(\param)$ and 
$\sigma^2(\param)$ in terms of samples already explored in the parameter
space. In a GP model, at sampled parameter values, the variance is
zero, but at points that are away from the sampled values, the variance
could be high. We now give the symbolic expressions for the mean
and variance of a GP model in terms of a kernel function 
$k(\param,\param)$. For ease of exposition let $\hat{\paramspace}$
denote the vector of parameter values already sampled. Then,
Let $Y$ denote the vector of robustness values for parameter values in
$\hat{\paramspace}$. Then, from \cite{gpbook}, Chapter 2, we have:
\begin{eqnarray}
    \mu(\param)    & = & k(\param,\hat{\paramspace})^{\top}\,k(\hat{\paramspace},\hat{\paramspace})^{-1}\, Y \\
    \Sigma(\param) & = & k(\param,\param) - 
                         k(\param,\hat{\paramspace})\,k(\hat{\paramspace},\hat{\paramspace})^{-1}\,
                         k(\hat{\paramspace},\param) \\ 
    \sigma(\param) &= & \sqrt{\Sigma(\param)}
\end{eqnarray}
The main idea is to use the mean and variance of the GP model to
prioritize searching parameter values where: (a) the robustness may 
be close to zero, (b) the variance of the GP model may be high. 
These two choices give us two different ways to partition the 
parameter region that we now explain. 

In the literature on GP-based Bayesian optimization, there is work
on defining {\em acquisition functions} that are used as targets for
optimization. Examples include UCB (Upper Confidence Bound) acquisition
function that is a combination of the mean and variance of the GP,
EI (Expected Improvement) which focuses on the expected value of the
improvement in the function value etc. Inspired by the UCB function
that allows a trade off between exploration and exploitation, we 
consider two acquisition functions: (1)  the first is focused on
pure exploration and uses the variance of the GP as the objective for
maximization, (2) the second is the difference between the mean and
the standard deviation. The rationale for the second function is 
that if $\mu(\param) - \sigma(\param)$ is lower than $0$, then it
is an indicator of a low robustness region. 

\begin{algorithm}[htb]
\caption{Parameter Space Partition using GP models $\mathtt{Partition}(S,\regalg)$}
\label{alg:partition}
\SetKwInOut{Input}{input}
\SetKwInOut{Output}{output}
\Input{Parameter set $S$, regression algorithm $\regalg$}
\Output{parameter sets $S_1,\ldots$ that will be pushed into the set $\paramspace^{r}$}
$f(\param)$ $\gets$ $\mathsf{acquisition}(\mu(\param),\sigma(\param))$ \;
$\param_1,\ldots,\param_k$ = $\{ \param \mid \psi(f(\param))\ \text{is true} \}$\;
\Return $\mathsf{Split}(S, \param_1, \ldots, \param_k)$ \;
\end{algorithm}

Algorithm~\ref{alg:partition} presents the new $\mathtt{Partition}$ function of Algorithm~\ref{alg:main} using GP-based acquisition.
In general, we assume a function $f$ that is an acquisition function,
and some property $\psi$ of the acquisition function. We pick points
$\param^*$ such that $\psi(f(\param))$ is true. For the first acquisition
function, $\psi$ is the property that chooses $\param^*$ that maximizes
the value of $f(\param)$. For the second acquisition function, we pick
$\psi$ as the property where $f(\param) \approx 0$, that is, based on
finding the roots of the function $f(\param)$. 
In Algorithm~\ref{alg:partition}, $\mathsf{Split}$ is a function that 
uses the points $\param_1$, $\ldots$, $\param_k$ to split $S$ into a
number of regions. The number of regions depends exponentially on the
dimension of the parameter space. For example, when we use the first
acquisition function, we use a single parameter value that maximizes the
uncertainty. If the dimension of the parameter space is $2$, then each
split point will generate $4$ new regions, and this is also called the
{\em greatest uncertainty split}. For the root-finding based procedure
(using  the second acquisition function), the number of regions depends on
the number of approximate roots of the acquisition function, we also 
call this the {\em root split} method. Essentially, the points identified
by $\psi(f(\param))$ become corner points of new regions.

\section{Case Studies}\label{sec:casestudies}



In this section, we present case studies of CPS models, and identify
regions in the parameter space that we can mark as safe, unsafe or
unknown with high probability. We tried each of the case studies with
different regression algorithms, with Gaussian Process regression
leading to smaller residuals, ergo, narrower conformal intervals. We
tried both (a) the na\"{i}ve algorithm that recursively splits the
parameter space, and (b) the algorithm which adaptively partitions the
parameter space exploiting the uncertainty as expressed by a Gaussian
Process prior.

For all case studies, we used a miscoverage level of $\epsilon = 0.05$
(i.e. providing a probability threshold of $95\%$). The GP regression
uses a {\em sum kernel}  $k(\param_i,\param_j) =
k_1(\param_i,\param_j) + k_2(\param_i,\param_j)$, where $k_1$ is the
dot product kernel, i.e., $k_1(\param_i,\param_j) = \sigma_0^2 +
\param_i \cdot \param_j$, and $k_2$ is the white kernel, where
$k_2(\param_i,\param_j) = 1$ if $\param_i = \param_j$ and $0$
otherwise. Table~\ref{tb:benchmark1} summarizes the performance of
Algorithm~\ref{alg:main} using Algorithm~\ref{alg:partition} with the
greatest uncertainty split method on all our benchmarks. Moreover,
slightly different from Algorithm~\ref{alg:main}, when a cover reaches
the lower size bound $\delta_{\min}$, it is marked unsafe if we can find
a counter-example inside it and unknown otherwise.

\begin{table}[t]
\centering
\scalebox{0.9}{
\begin{tabular*}{.99\textwidth}{@{\extracolsep{\fill}}llllll}
\toprule
Case Study &   \multicolumn{3}{c}{Ratio of Volumes (\%)} & Sims./ & Spec.\\
\cline{2-4}
           & Safe & Unsafe &  Unk.   & region  & \\ 
\midrule
Mountain Car 1 & 88.72 & 11.28 & 0.00 & 100 &  $\specmc$  \\
  
Lane Keep Assist 1 & 100 & 0.00 & 0.00 & 100 &  $\speclkabounds$ \\
Lane Keep Assist 2 & 77.23 & 21.97 & 0.80 & 100 & $\speclkasettle$ \\
F16 Level Flight & 67.18 & 32.81 & 0.00 & 100 &  $\phi_{\textsc{f16,level}}$ \\
F16 Pull up & 43.52 & 56.09 & 0.40 & 100 &  $\phi_{\textsc{f16,pullup}}$ \\
F16 GCAS & 3.91 & 96.09 & 0.00 & 100 & $\phi_{\textsc{f16,GCAS}}$ \\



Simglucose 2M & 45.45 & 54.55 & 0.00 & 10 & $\psi_{\textsc{simglucose,2M}}$ \\
Simglucose 2M & 100 & 0.00 & 0.00 & 10 & $\phi_{\textsc{simglucose,2M}}$ \\

Simglucose 3M & 100 & 0.00 & 0.00 & 10 & $\phi_{\textsc{simglucose,3M}}$ \\
Simglucose 4M & 100 & 0.00 & 0.00 & 10 & $\phi_{\textsc{simglucose,4M}}$ \\
Simglucose 5M & 100 & 0.00 & 0.00 & 10 & $\phi_{\textsc{simglucose,5M}}$ \\
\bottomrule

\end{tabular*}
}
\caption{Performance of Algorithm~\ref{alg:main} using the GP-based greatest uncertainty split method with 95\% confidence level}
\label{tb:benchmark1}
\end{table}

\paragraph{Mountain Car.}
A classic problem used in the RL literature the mountain car models an under-powered car attempting to drive up a hill. To successfully climb the steep hill, the car needs to accumulate potential energy by going in the opposite direction and then use the gained momentum to power over the hill. The dynamics of the mountain car~\cite{zareistatistical} are described using $\dot{x} = v, \dot{v} = - mg \cos(3x) + \frac{F}{m}u - \mu \dot{x}$,
where $x$, $v$, $m = 0.2$ kg, $g = 9.8
ms^{-2}$, $f = 0.2$ N, and $\mu = 0.5$ are respectively the position
(in meters), velocity (in $ms^{-1}$), mass, acceleration due to
gravity, force produced by the car engine, and the friction factor
respectively. The input from the controller is denoted by $u(t)$.  The
parameter space is defined by the initial position $\initPos$ and
velocity $\initVel$ of the car. 
We wish to identify regions of space that
satisfy or violate the property of reaching the
goal. The region we choose for analysis is defined as as $\paramspace$ = $(\initPos,\initVel) \in [-0.7,0.2]\times[-0.5,0.5]$, which is 
comparable to the region used in \cite{zareistatistical}. 
We consider the initial parameter setting safe if it satisfies
the STL formula $\specmc = \eventually_{[0, 10]} (x(t) > 0.45).$

\begin{figure}[t]\centering
\subfloat[Parameter space partitioning obtained using naive split ($\delta_{\min} = 0.02$)]{
\includegraphics[width=0.45\linewidth]{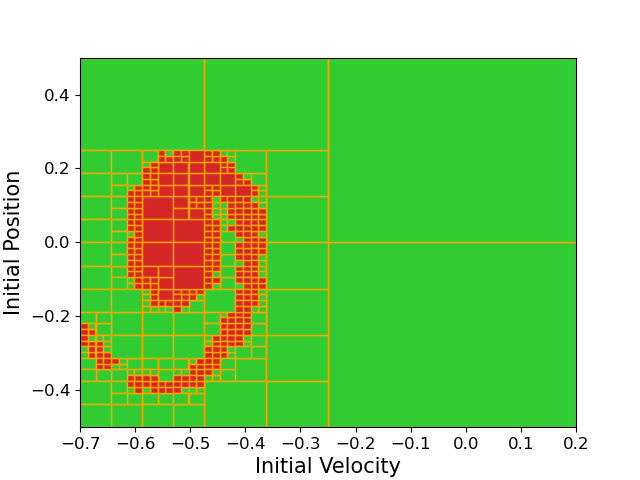}\label{mc_regression_naive}}
\subfloat[Parameter space partitioning obtained using root split ($\delta_{\min} = 0.02$)]{
\includegraphics[width=0.45\linewidth]{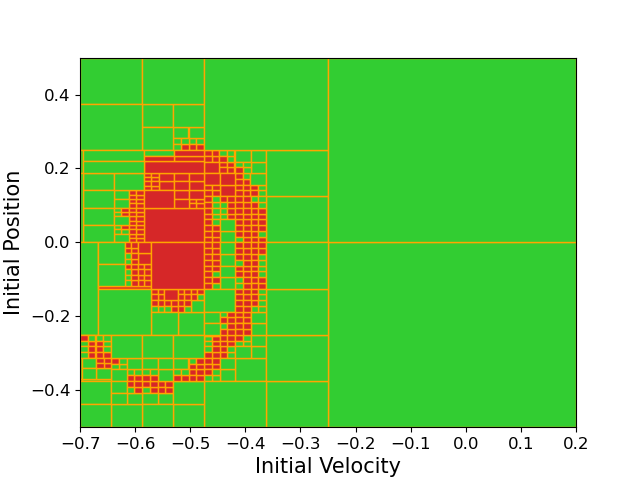}\label{mc_regression_zero_points}}
\hfill
\subfloat[Parameter space partitioning obtained using greatest uncertainty split ($\delta_{\min} = 0.02$)]{
\includegraphics[width=0.45\linewidth]{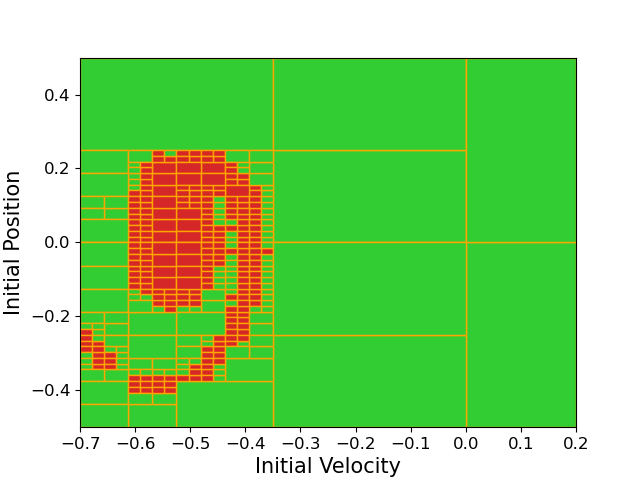}\label{mc_regression_ucb}}
\subfloat[{Mountain Car Ground Truth}]{
\includegraphics[width=0.45\linewidth]{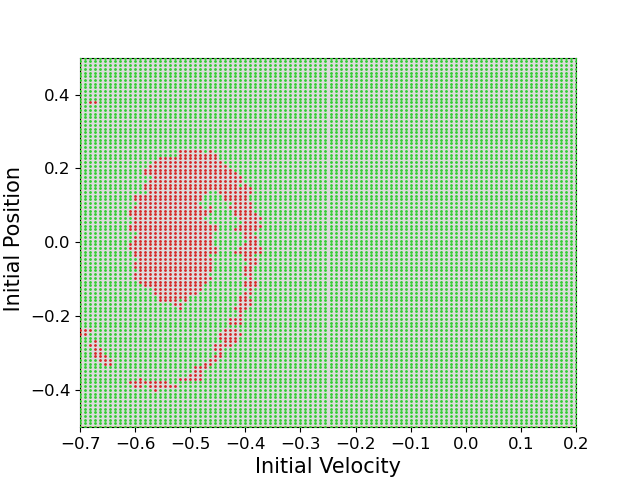}\label{mc_gt_dot}}
\caption{\small Mountain Car parameter space partitioning using different approaches}
\label{figure:mc}
\end{figure}

\begin{wraptable}{r}{0.6\textwidth}
\centering
\scalebox{0.8}{
\begin{tabular}{@{\extracolsep{\fill}}lclll}
\toprule
$\mathtt{Partition}$ method & \#  covers &   \multicolumn{3}{c}{Ratio of Volumes (\%)} \\
\cline{3-5}
       &    & Safe & Unsafe &  Unk.   \\ 
\midrule
Na\"{i}ve (Section~\ref{sec:naive}) & 457 & 89.01 & 10.99 & 0.00\\
Root split &  481 & 89.23 & 10.77 & 0.00\\
Greatest uncertainty & 364 & 88.72 & 11.28 & 0.00 \\
\bottomrule

\end{tabular}
}
    \caption{\footnotesize Comparison  of Algorithm~\ref{alg:main} using different $\mathtt{Partition}$ approaches with 95\% confidence level.}
    \label{tab:diff-part}
\end{wraptable}
We compare with approximate ground truth obtained by uniform sampling of the parameter
space in Fig.~\ref{mc_gt_dot}. Green dots indicate parameter values that lead
to satisfaction of the property, while red dots indicate violations. In 
Fig.~\ref{mc_regression_naive}~\ref{mc_regression_zero_points}~\ref{mc_regression_ucb}, we show the results obtained using
Algorithm~\ref{alg:main}. Table~\ref{tab:diff-part} summaries the parameter space partitioning results where green regions correspond to parameters that leads to satisfaction of $\specmc$, and red regions correspond to parameters that leads to violation of $\specmc$. Although the partition covers look different using the different partitioning strategies, the safe/unsafe ratio remain similar and they all provide the same level of probabilistic guarantees. Among the three partitioning approaches, we have the least number of total covers using the greatest uncertainty split method. In general, less number of covers lead to less number of simulations. 




\paragraph{Reinforcement Learning Lane Keep Assist.}

Lane-keep assist (LKA) is an automated
driver assistance technique used in semi autonomous vehicles to keep
the ego vehicle traveling along the centerline of a lane. The recent
reinforcement learning toolbox from \matlab introduces a Deep Q
Network (DQN)-based reinforcement learning agent that seeks to keep
the ego vehicle centered. The inputs to the RL agent are the lateral
deviation $e_1$, relative yaw angle (i.e. yaw error)
$e_2$, their derivatives and their integrals.  Specifics of the
DQN-based RL agent and its training can be found in \cite{rllka}. The
parameter space for this model is the initial values for $e_1$ and
$e_2$, where we looked at region $\paramspace=(e_1, e_2) \in [-0.3, 0.3]\times[-0.2, 0.2]$. We are interested in the signals corresponding to the lateral
deviation and yaw error. We are interested in checking the
control-theoretic properties such as overshoot/undershoot bounds and
the settling time for these signals. In this experiment, we consider
two properties characterizing bounds on $e_2$ and settling time for
$e_1$; $\speclkasettle:\always_{[2,15]}(|e_1| < 0.025)$ and $\speclkabounds:\always_{[0,15]}(e_2 < 0.4 \wedge e_2 > -0.4)$.

Figure~\ref{figure:lka} shows the parameter space partitioning results and the ground truth with respect to
$\speclkasettle$. Our technique was able to certify that $\speclkabounds$ is satisfied by the entire region with 95\% confidence.


\begin{figure}[htb]\centering
\subfloat[Parameter space partitioning using greatest uncertainty split ($\delta_{\min}=0.01$)]{
\includegraphics[width=0.45\linewidth]{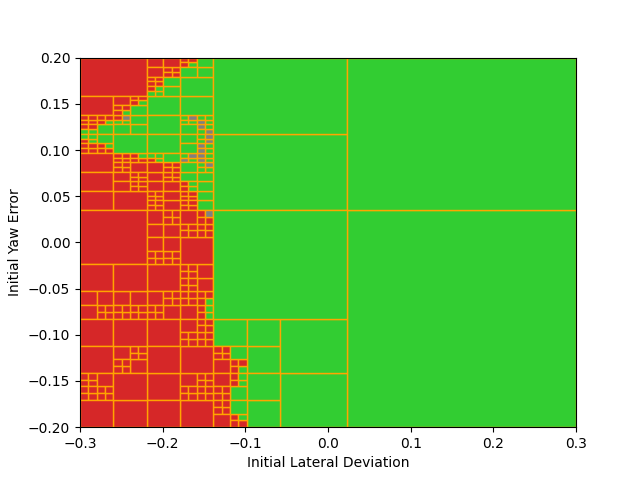}\label{lka_regression}}
\subfloat[{LKA (settle) Ground Truth}]{
\includegraphics[width=0.45\linewidth]{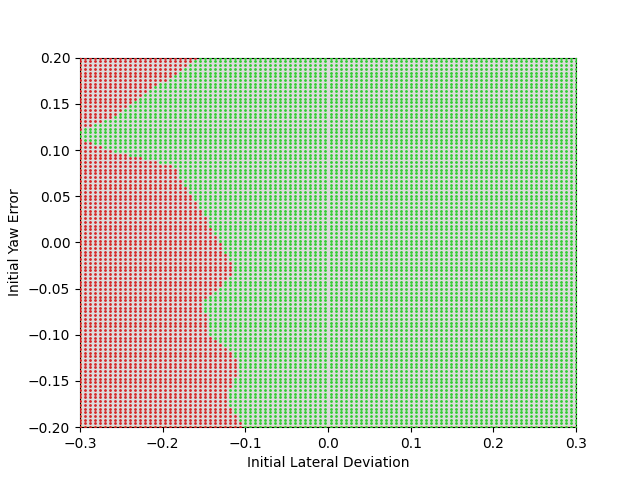}\label{lka_gt_dots}}
\caption{\small Lane Keep Assist }
\label{figure:lka}
\end{figure}




\paragraph{F-16 Control System.}
We now present the results of our technique on the verification challenge~\cite{heidlauf2018verification} which requires the analysis of the F-16 flight control loop system. The system is modeled as a hierarchical control system - an outer-loop autopilot and an inner loop tracking and stabilizing controller (ILC), and a $13$ dimensional non-linear dynamical plant model.
The plant dynamics are based on a 6 degrees of freedom standard airplane model~\cite{stevens2015aircraft} represented by a system of $13$ ODEs describing the force equations, kinematics, moments and a first-order lag model for the afterburning turbofan engine. These ODEs describe the evolution of the system states, namely velocity $vt$, angle of attack $\alpha$, sideslip $\beta$, altitude $h$, attitude angles: roll $\phi$, pitch $\theta$, yaw $\psi$, and their corresponding rates $p$, $q$, $r$, engine $power$ and two more states for translation along north and east. The non-linear plant model uses linearly interpolated lookup tables to incorporate wind tunnel data.
The control system is composed of an autopilot that sets the references on upward acceleration, stability roll rate and the $throttle$. The ILC uses an LQR state feedback law to track the references and computes the control input for the $aileron$, $rudder$ and the $elevator$.
We now present our results on three separate scenarios capturing specific contexts. Each scenario defines the parameter set and an associated specification.




\begin{figure}[t]\centering

\subfloat[Parameter space partitioning using greatest uncertainty split ($\delta_{\min} = 0.1$)]{
\includegraphics[width=0.44\linewidth]{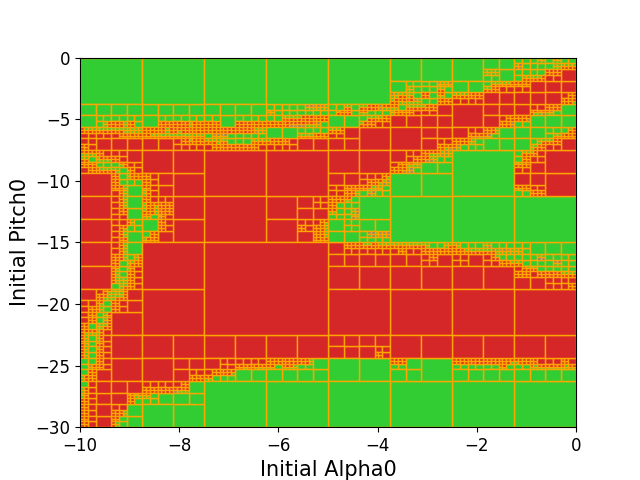}\label{pullup_regression}}
\subfloat[{Ground truth}]{
\includegraphics[width=0.44\linewidth]{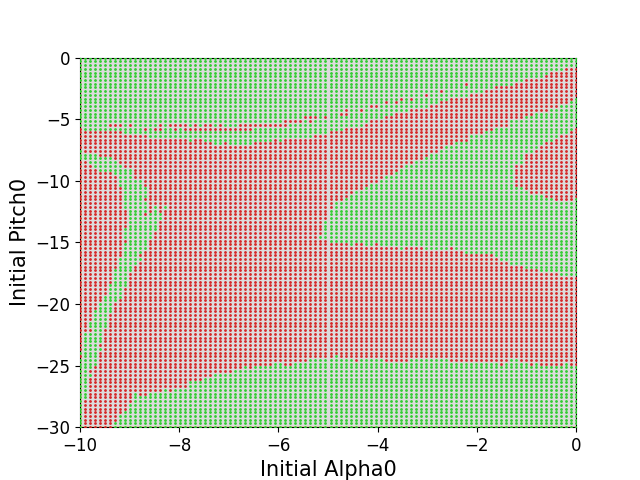}\label{pullup_gt_dot}}
\hfill
\subfloat[Parameter space partitioning using greatest uncertainty split ($\delta_{\min} = 100$)]{
\includegraphics[width=0.44\linewidth]{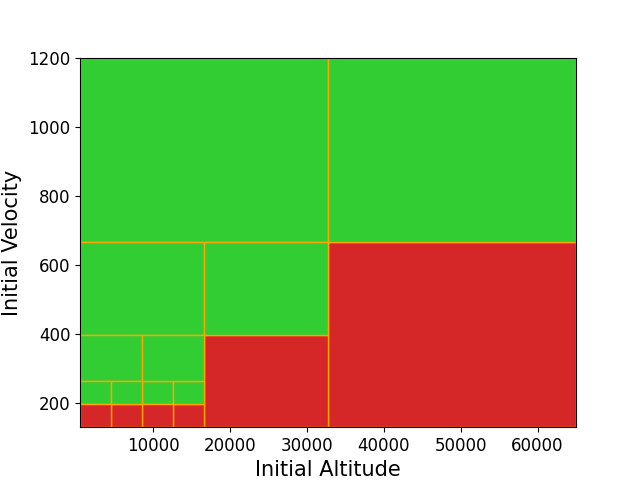}\label{f16_regression_level}}
\subfloat[{Ground truth}]{
\includegraphics[width=0.44\linewidth]{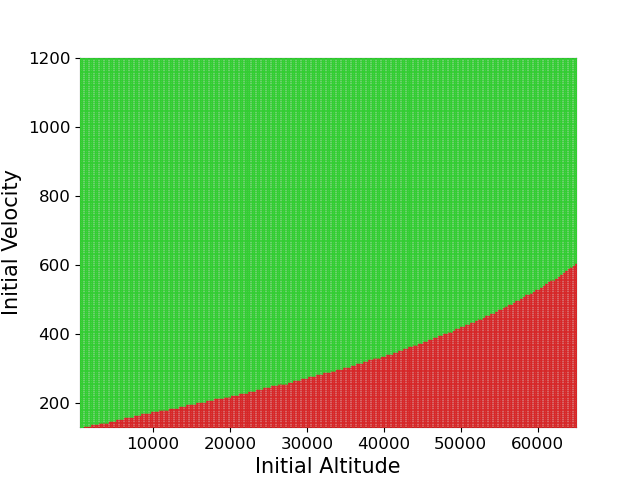}\label{f16_level_gt}}

\caption{\small F16 - Level Flight (bottom), Pull up maneuver (top)}
\label{figure:f16_level}
\end{figure}


\textbf{F16-Pull up maneuver}. This scenario demonstrates the tracking of a constant autopilot command requesting an upward acceleration ($N_z=5g$). The ILC tries to track the reference without undesirable transients like pitch oscillations and exceeding pitch rate limits. We modify the controller gains to highlight the violations of the spec $\phi_{\textsc{f16,pullup}}: \always_{[0,10]}q \leq 120deg/s$. The bounded parameter space is described by initial values of $\alpha\in[-10,0]deg, \theta\in[-30,0]deg$ and the results are shown in Figure~\ref{pullup_regression},~\ref{pullup_gt_dot}.

\textbf{F16-Level Flight}. This scenario describes straight and level flight with a constant attitude and $0$ initial angular rates. The bounded parameter space is defined by the initial altitude $h\in[500, 65000]$ and velocity $vt\in[130, 1200]$. The autopilot references are set to zero, and the ILC tries to maintain a constant altitude and angle of attack $\alpha$. As the F-16 can fly over a large range of altitudes and velocities, a single LQR computed against the linearzied model can not satisfy the goal and results in a stall defined by $\phi_{\textsc{f16,level}}: \always_{[0,10]}\alpha \leq 35deg$. This is reflected in the results as shown in Figure~\ref{figure:f16_level}.




\textbf{F16-Ground Collision Avoidance(GCAS)}. The final scenario describes the F-16
diving towards the ground and the GCAS autopilot trying to prevent the
collision. The GCAS brings the roll angle and its rate to $0$ and then accelerates upwards to avoid ground collision as defined by the spec 
$\phi_{\textsc{f16,gcas}}: \always_{[0,10]}h \geq 0ft$. The parameter
space is described by initial values of $\alpha\in[0.075, 0.1]rad$ and 
$\phi\in[-0.1, -0.075]rad$ . The results are illustrated in Figure~\ref{figure:f16_gcas}. 
In this case study, the ground truth and our results seem to be less well-matched
than other case studies. There are a couple of reasons for this. First,
observe that the ground truth is highly non-monotonic. Given the nonlinearity
of the ground truth, our conformal interval prediction errs on the side of
safety and marks a region unsafe. To remedy this, we would have to increase the 
number of simulations per region used to train the GP regression model
and possibly experiment with other regression models (such as a deep neural 
network regressor). 




\paragraph{Artificial Pancreas.}

Type-1 diabetes (juvenile diabetes) is a chronic condition caused by the inability of the pancreas to secrete the required amount of insulin.
Simglucose~\cite{simglucose2018} is a Python implementation of the FDA-approved Type-1 Diabetes simulator~\cite{man2014uva} which models glucose kinetics. We input a list of tuples of time and meal size to Simglucose and set the same scenario environment. Choosing patients in different age will result in different simulation trace. The parameter \textit{meal time} is constrained to be strictly increasing and the last meal time to be less than 24. For each scenario, the simulator provides traces records for different blood indicators based on a given setting environment. We are interested in checking the property of blood glucose(BG), such as if the blood glucose will violate an upper limitation in a day. In this case study, we implement our technique on this system verification by setting alpha as 0.05 and delta as 0.5. Therefore, we present the results with 95\% confidence guarantee and any subset size less than 0.5 is considered as an unknown region.

We study $4$ scenarios (Table~\ref{tab:simglucose}) describing an adolescent patient who takes $2$,$3$,$4$ and $5$ meals a day respectively. The meals of size $s_i$ are consumed at time $t_i$. The parameters space is then defined by $\paramspace=(t_1, s_1,\ldots, t_i, s_i,\ldots t_n, s_n)$, where $n$ is the total number of meals. The properties $\chi_{\textsc{simglucose},n\textsc{M}}: \always_{[0,T]}BG \leq UP$ (where $\chi\in{\phi,\psi}$), asserts that the patient's blood glucose should not go beyond the upper bound ($UP$) in time $T$. Our results predict the whole region as $100\%$ safe region with $95\%$ confidence. Besides, for the two meal case: $\psi_{\textsc{simglucose,2M}}$, our implementation result of 54.44\% matches well based on the ground truth dataset with 52.47\% unsafe volume.

\begin{table}[t]
\centering
\begin{tabular}{@{\extracolsep{\fill}}lr}
\toprule
Specification & Unsafe Volume (Ground Truth) (\%) \\
\midrule
\multicolumn{2}{l}{\small$\paramspace_{2M}\in [1, 12]\times[1, 20]\times[13, 24]\times[1, 20]$}\\
$\psi_{\textsc{simglucose,2M}}: \always_{[0,24]}BG \leq 155$ & 52.47\\
$\phi_{\textsc{simglucose,2M}}: \always_{[0,24]}BG \leq 170$ & 0\\
\midrule
\multicolumn{2}{l}{\small$\paramspace_{3M} \in [1, 8]\times[1, 20]\times[9, 16]\times[1, 20]\times[17,24]\times[1, 20]$}\\
$\phi_{\textsc{simglucose,3M}}: \always_{[0,24]}BG \leq 170$ & 1.2 \\
\midrule
\multicolumn{2}{l}{\small$\paramspace_{4M} \in [1, 6]\times[1, 20]\times[7, 12]\times[1, 20]\times[13,18]\times[1, 20]\times[19,24]\times[1,20]$}\\
$\phi_{\textsc{simglucose,4M}}: \always_{[0,24]}BG \leq 170$ & 4.7\\
\midrule
\multicolumn{2}{l}{\tiny$\paramspace_{5M} \in [1, 4]\times[1, 20]\times[5, 9]\times[1, 20]\times[10, 14]\times[1, 20]\times[15, 19]\times[1,20]\times[20, 24]\times[1,20]$}\\
$\phi_{\textsc{simglucose,5M}}: \always_{[0,24]}BG \leq 170$ & 6.6\\
\bottomrule

\end{tabular}
    \caption{\footnotesize Simglucose results with 95\% confidence level for different parameter spaces and specifications.}
    \label{tab:simglucose}
\end{table}

\myipara{Impact of Results} The parameter space partitioning produced by our algorithm helps characterize safe and unsafe regions of operation for CPS applications. This helps in selection of the appropriate model parameters or 
initial conditions during design time and runtime. Furthermore,
if the model parameters are being chosen by an outer loop supervisory control,
then the partitions that we generate create conditional contracts on the 
safety of the CPS model. Ultimately, we envision the results of our tool
to be used for constructing safety assurance cases\cite{rushby2002partitioning}.

\section{Related Work and Conclusions}\label{sec:relconc}
\mypara{Related Work}
Several verification techniques for CPS have been proposed in the recent years. Some are specifically developed to analyze models of neural networks (NN) or NN-controlled dynamical systems in order to obtain safety guarantees, using barrier functions~\cite{Royo2018ClassificationbasedAR,JimJyo}, reachability analysis~\cite{Taylor,huang2017safety,Verisig,fazlyab2019efficient}, Satisfiability Modulo Theory (SMT) based methods or mixed-integer linear program (MILP) optimizer~\cite{Reluplex,Yasser,dutta2018learning}. Most of these methods provide deterministic guarantees but face scalability issues and are restricted by the class of models they can handle (e.g. ReLU activation functions). Unlike these approaches, our method is applicable to a broader class of CPS as it treats learning-enabled CPS as black-boxes and only assumes that they can compute simulations given fixed parameters. Moreover, we reason over the joint distribution of the parameters and corresponding robust values with respect to STL properties. Therefore, our algorithm is independent of the model complexity (except for the computational complexity of the simulations) and more scalable.


Methods based on Statistical Model Checking (SMC)~\cite{SMCOverview,legay2015statistical,zuliani2010bayesian} can overcome the hurdles like scalability and nonlinearity and provide probabilistic guarantees~\cite{zareistatistical,NimaPowerTrain,YuSMCCPS,clarke2008statistical,abbas2014robustness}. These methods are based on statistical inference methods like sequential probability ratio tests~\cite{SMCOverview,NimaPowerTrain,MaheshSMC,clarke2008statistical}, Bayesian statistics~\cite{zuliani2010bayesian}, and Clopper-Pearson bounds~\cite{zareistatistical}. Another line of works use Probably Approximately Correct (PAC) learning theory to give probabilistic bounds for Markov decision processes and black-box systems~\cite{fu2014probably,Fan2017Dryvr:Systems}. In contrast to SMC and PAC-learning techniques, our approach is sample independent and can provide the required probabilistic guarantees with any number of samples. This is because we build a guaranteed regression model from the system parameters with respect to the robust satisfaction value of the corresponding STL properties. If the regression model is of poor quality (due to few samples), using the calibration step in conformal regression, the predicted (but wider) interval can
still have the same level of guarantee. Conformal regression lets us tradeoff the quality of the regression model (w.r.t. the data) and the width of the interval for which we have high-confidence property satisfaction, and {\bf not} the level of the guarantee itself. 

\mypara{Conclusions}
In this paper, we proposed a  verification framework that can search the parameter space to find the regions that lead to satisfaction or violation of given specification with probabilistic coverage guarantees. There are a couple of directions we aim to explore as future work: 1) We used a very basic version of conformal regression in Algorithm~\ref{alg:conformal}, which gives a constant confidence range $d$ across all $X$. Techniques based on quantile regression~\cite{romano2019conformalized} and locally-weighed conformal~\cite{lei2018distribution} can make $d$ a function of $X$ and give much shorter prediction intervals. 2) We plan to explore probabilistic regret bounds for Gaussian process optimization to help 
obtain (probabilistic) upper and lower bounds on the value of the
surrogate model when using GP-based regression. 

\bibliographystyle{splncs04}
\bibliography{references}

\clearpage
\appendix
\section*{Appendix}

%

\begin{figure}[h!]\centering
\includegraphics[height=5cm]{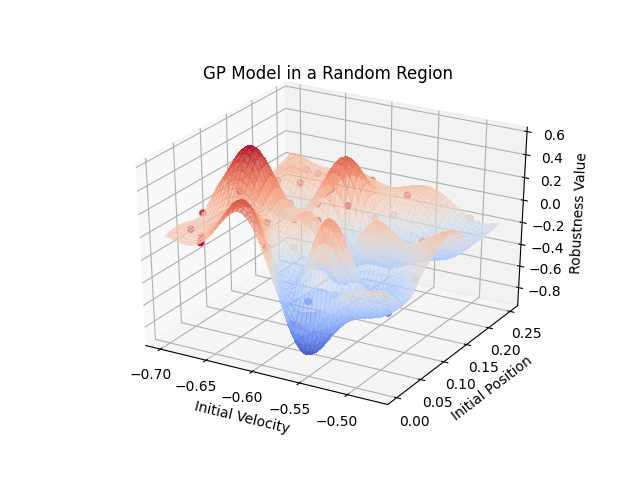}
\caption{\small Gaussian Process Regression Visualization for the Mountain Car model}
\label{figure:gp_model}
\end{figure}

\begin{figure}[h!]\centering
\subfloat[Parameter space partitioning obtained using greatest uncertainty split ($\delta_{\min}$=0.002)]{
\includegraphics[width=0.5\linewidth]{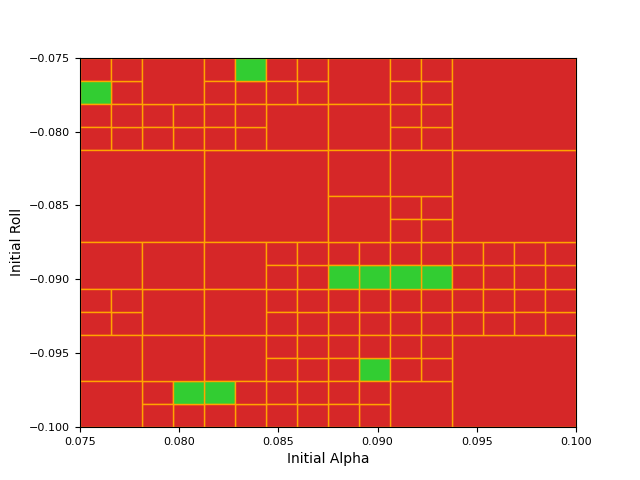}\label{gcas_regression}}
\subfloat[{GCAS Ground truth}]{
\includegraphics[width=0.5\linewidth,height=0.14\textheight]{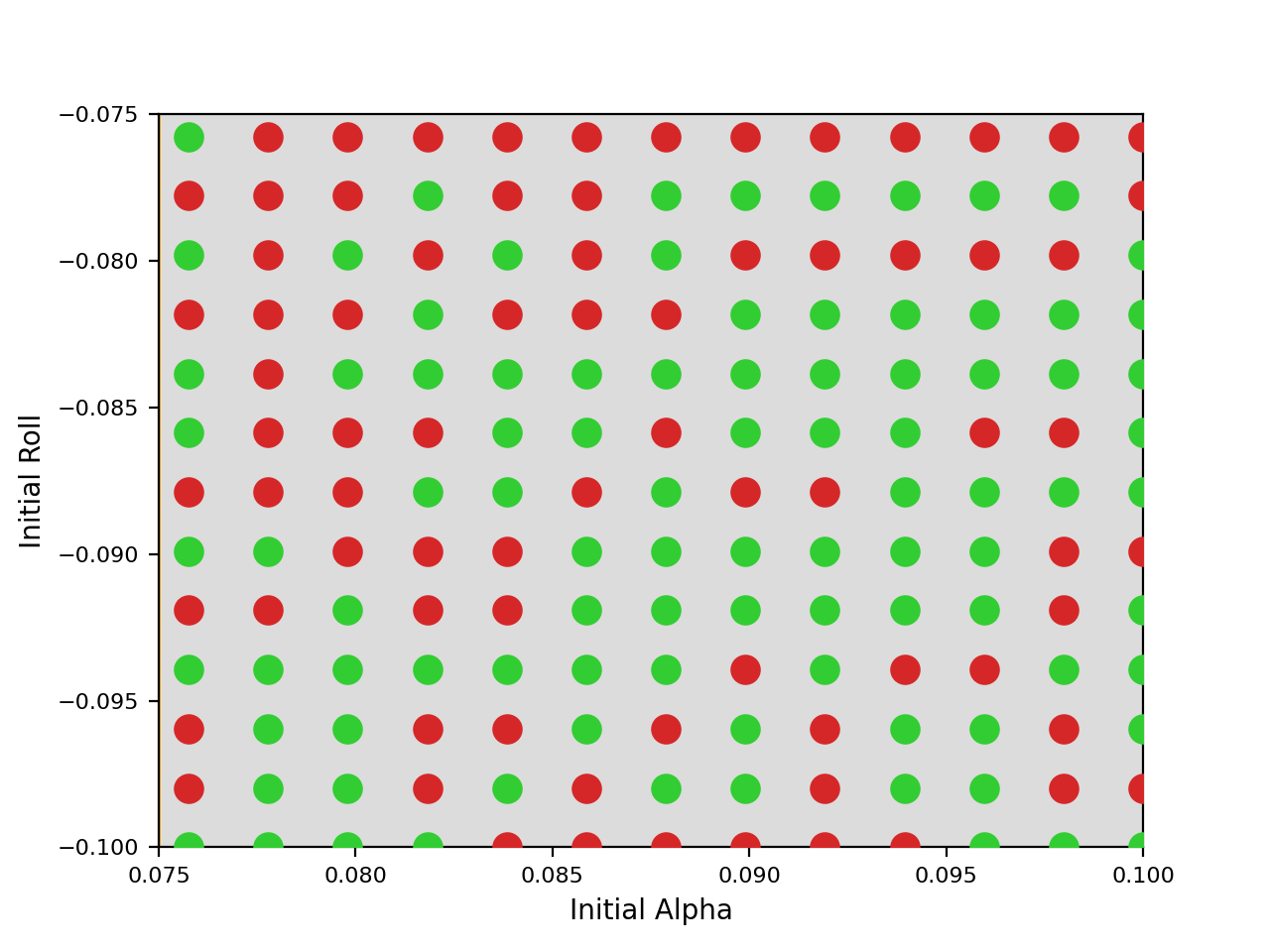}\label{gcas_gt_dot}}
\caption{\small F16 - Ground Collision Avoidance}
\label{figure:f16_gcas}
\end{figure}
%


\end{document}